\documentclass[conference,10pt]{IEEEtran}
\usepackage{cite}

\ifCLASSINFOpdf
   \usepackage[pdftex]{graphicx}

\else

\fi

\bstctlcite{IEEEexample:BSTcontrol}

\usepackage{cite}

\usepackage{amsmath,amssymb,amsthm}

\usepackage{mathtools}

\usepackage{dblfloatfix}
\usepackage[caption=false,font=footnotesize]{subfig}
\usepackage{mathrsfs}
\newtheorem{theorem}{Theorem}

\newtheorem{lemma}[theorem]{Lemma}

\newtheorem{corollary}[theorem]{Corollary}

\begin{document}

\title{On Fractional Linear Network Coding Solution of Multiple-Unicast Networks}
\author{\IEEEauthorblockN{Niladri Das and Brijesh Kumar Rai}
\IEEEauthorblockA{Department of Electronics and Electrical Engineering\\ Indian Institute of Technology Guwahati, Guwahati, Assam, India\\
Email: \{d.niladri, bkrai\}@iitg.ernet.in}}
\maketitle

\begin{abstract}
It is known that there exists a multiple-unicast network which has a rate $1$ linear network coding solution if and only if the characteristic of the finite field belongs to a given finite or co-finite set of primes. In this paper, we show that for any non-zero positive rational number $\frac{k}{n}$, there exists a multiple-unicast network which has a rate $\frac{k}{n}$ fractional linear network coding solution if and only if the characteristic of the finite field belongs to a given finite or co-finite set of primes.
\end{abstract}

\section{Introduction}
The concept of network coding came into light in the year 2000 by a seminal work of Ahlswede \textit{et al.}. In \cite{ahlswede} the authors showed that the min-cut bound in a multicast network can be achieved if the intermediate nodes are allowed to do operations on incoming messages before forwarding. This kind of operation of intermediate nodes has been termed as network coding. Linear network coding refers to the scheme when all such operations are linear. In subsequent works it has been shown that linear network coding is sufficient to achieve the capacity of multicast networks \cite{li,medrad,jaggi}. Moreover, a capacity achieving network code can be designed efficiently \cite{jaggi}. As far as a multicast network is concerned, a scalar linear network code over a sufficiently large finite field suffices to achieve the capacity. However, the requirement of field size can be reduced if vector linear network coding is used \cite{jaggi2}. In \cite{riis}, it was shown that binary field is sufficient if the vector length is large enough. These results related to multicast networks were known by the year 2005. After a decade later, there has been a few interesting works related to intricacies involved in linear network coding for multicast network. In one such result it has been shown that a multicast network being linearly solvable over a sufficiently large finite field does not necessarily mean that over all larger fields it is also solvable \cite{sun1}. It was also shown that existence of a vector linear solution for a certain vector dimension does not necessarily mean that there exists a vector linear solution for all larger vector dimensions \cite{sun2}. It is important to note that for a multicast network, the characteristic of the finite field does not play an important role in the sense that there does not exist a multicast network which has a scalar/vector linear solution only over a finite field of certain characteristic.

The non-multicast networks have very different properties. For non-multicast networks it has been shown that there exists a network where better throughput (capacity) can be achieved if the nodes are allowed to do non-linear operations \cite{insuff}. There are non-multicast networks where scalar/vector linear solution not only depends on the size of the field but also on the characteristic of the field \cite{poly, rai1, rai2}. In particular, it has been shown in \cite{poly} that for any system of polynomial equations over integers, there exists a network which has a scalar linear network solution over a finite field if and only if the system of polynomial equation has a root in the same finite field. Therefore, there exist networks which are solvable only over fields of certain characteristics. Interesting examples are the Fano and non-Fano network presented in \cite{insuff,matroid}. The Fano network has a vector linear solution for any vector dimension if and only if the characteristic of the finite field is even. Over a finite field of odd characteristic it has linear coding capacity equal to $\frac{4}{5}$. The non-Fano network has a vector linear solution for any vector dimension if and only if the characteristic of the finite field is odd. If the characteristic is even then its linear coding capacity is equal to $\frac{10}{11}$.

In the references \cite{rai1} and \cite{rai2} the authors considered networks where all terminals demand the sum of the symbols generated at the sources. Such networks have been given the name sum-networks. In \cite{rai1} it was shown that for any finite set of primes there exists a sum-network which has a vector linear solution for any vector dimension if and only if the characteristic of the finite field belongs to the given set. In \cite{rai2} it was shown that for any sum-network, there exists a multiple-unicast network which has a rate $1$ solution if and only if the sum-network has a rate $1$ solution. They also show that for any co-finite set of primes there exists a sum-network which has a vector linear solution for any vector dimension if and only if the characteristic of the finite field belong to the given set \cite{rai2}. Thus, these three results when combined, shows that for any given finite/co-finite set of primes there exists a multiple-unicast network which has a vector linear solution if and only if the characteristic of the finite field belong to the given set.

However, in the works of \cite{poly, rai1, rai2}, the dependency on the characteristic of the field is shown only for either scalar linear network coding \cite{poly} or vector linear network coding of any vector length \cite{rai1,rai2}. In this paper, we generalise the previous results to show that for any non-zero positive rational number $\frac{k}{n}$ and for any given finite/co-finite set of prime numbers, there exists a non-multicast network which has a rate $\frac{k}{n}$ fractional linear network code solution if and only of the characteristic of the finite belongs to the given finite/co-finite set of primes.

The organization of the paper is as follows. In Section~\ref{sec2} we reproduce the standard definitions of fractional linear network coding, vector linear network and scalar linear network coding. In Section~\ref{sec3} we show that for any non-zero positive rational number $\frac{k}{n}$, and for any finite/co-finite set of primes, there exists a network which has a rate $\frac{k}{n}$ fractional linear network coding solution if and only if the characteristic of the finite field belongs to the given set. In Section~\ref{sec4} we extend the results of Section~\ref{sec2} to multiple-unicast networks. The paper is concluded in Section~\ref{sec5}.

\section{Preliminaries}\label{sec2}
A network is represented by a graph $G(V,E)$. The set $V$ is partitioned into three disjoint sets namely, the set of sources $S$, the set of terminals $T$, and the rest of the nodes are called the intermediate nodes and their collection is denoted by $V^\prime$. Without loss of generality the sources are assumed to have no incoming edge and the terminals are assumed to have no outgoing edge. Each source generates an i.i.d random process uniformly distributed over an alphabet $\mathcal{A}$. The source process at any source is independent of all source processes generated at other sources. Each terminal demands to compute the information generated at a subset of the sources. An edge $e$ originating from node $u$ and ending at node $v$ is denoted by $(u,v)$; where $u$ is denoted by $tail(e)$, and $v$ is denoted by $head(e)$. For an node $v\in V$, the set of edges $e$ for which $head(e) = v$ is denoted by $In(v)$. The information carried by an edge $e$ is denoted by $Y_e$. Without loss of generality it is assumed that all the edges in the network are unit capacity edges.

In a $(k,n)$ fractional linear network code the alphabet $\mathcal{A}$ is taken as a finite field $\mathbb{F}_q$. Each source $s_i\in S$ generates a symbol $X_i$ from the finite field $\mathbb{F}_q^k$. For any edge $e$, if $tail(e) = s_i$, then $Y_e = A_{\{s_i,e\}}X_i$ where $Y_e \in \mathbb{F}_q^n$, $A_{\{s_i,e\}} \in \mathbb{F}_q^{n\times k}$ and $X_i\in \mathbb{F}_q^k$. If $tail(e)$ is any intermediate node $v\in V^{\prime}$, then $Y_e = \sum_{\forall e^\prime \in In(v)} A_{\{e^\prime,e\}}Y_{e^\prime}$ where $Y_e,Y_{e^\prime} \in \mathbb{F}_q^n$, and $A_{\{e^\prime,e\}} \in \mathbb{F}_q^{n\times n}$. For any terminal $t\in T$, if $t$ computes symbol $X_t$, then $X_t = \sum_{\forall e^\prime \in In(t)} A_{\{e^\prime,t\}}Y_{e^\prime}$ where $X_t\in \mathbb{F}_q^k, A_{\{e^\prime,t\}} \in \mathbb{F}_q^{k\times n}$ and $Y_{e^\prime}\in \mathbb{F}_q^n$. The matrices $A_{\{s_i,e\}}, A_{\{e^\prime,e\}}$ and  $A_{\{e^\prime,t\}}$ shown above are called as the local coding matrices.

If using a $(k,n)$ fractional linear network code all terminals can compute $k$ respective symbols in $n$ uses of the network, then the network is said to have a $(k,n)$ fractional linear network coding solution. The ratio $\frac{k}{n}$ is called the rate. A network is said to have a rate $\frac{k}{n}$ fractional linear network coding solution if it has a $(dk,dn)$ fractional linear network coding solution for any non-zero positive integer $d$. A $(k,k)$ fractional linear network code is called as a $k$ dimensional vector linear network code and $k$ is called the vector dimension or as the message dimension. If a network has a $(k,k)$ fractional linear network coding solution then it is said that the network has a vector linear solution for $k$ message dimension. If a network has a $(1,1)$ vector linear network coding solution then the network is said to be scalar linearly solvable.

\section{A network having a rate $\frac{k}{n}$ fractional linear network coding solution iff the characteristic belongs to a given finite/co-finite set of primes}\label{sec3}

\subsection{Network having $\frac{k}{n}$ solution iff characteristic belongs to a given finite set of primes.}\label{subsec1}
\begin{figure*}
\centering
\includegraphics[width=.96\textwidth]{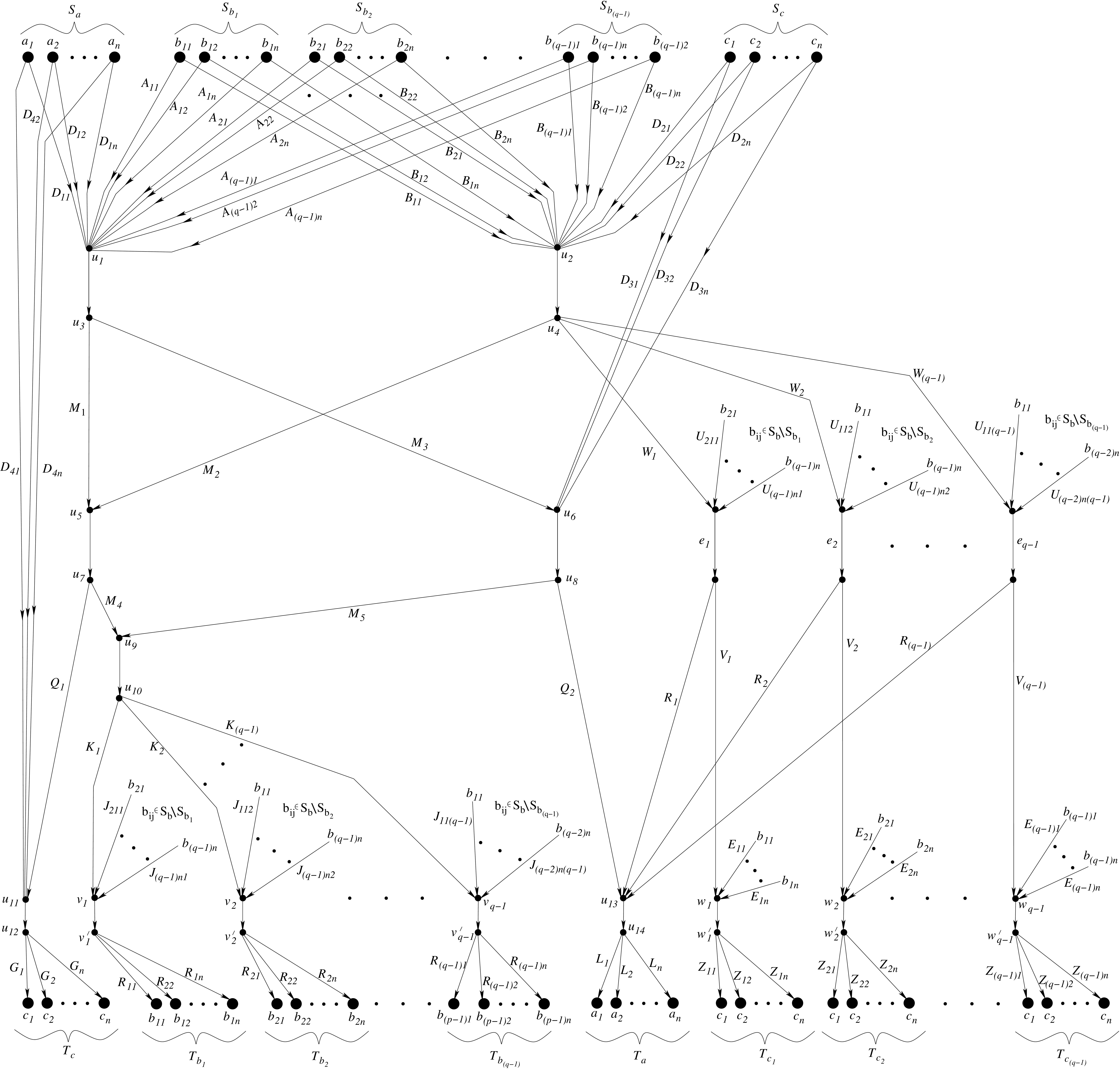}
\caption{Network $\mathcal{N}_1$ which has a rate $\frac{1}{n}$ fractional linear network coding solution if and only if the characteristic of the finite field divides $q$}
\label{genfano1/n}
\end{figure*}
First we show that for any positive non-zero rational number $\frac{k}{n}$, and for any given finite set of primes, there exists a network which has a rate $\frac{k}{n}$ fractional linear network coding solution if and only if the characteristic of the finite field belongs to the given set. Towards this end, we consider the network $\mathcal{N}_1$ presented in Fig.~\ref{genfano1/n}. The network shown here in Fig.~\ref{genfano1/n} has both the generalized Fano network shown in \cite{gf} and the Fano network shown in \cite{gf} as a sub-network. As it can be seen, the network has $q+1$ sets of sources namely, $S_a$, $S_{b_i}$ for $1\leq i\leq (q-1)$ and $S_c$. In the figure, the individual source nodes are indicated by the source message it generates. A source $s_i\in S_a $ generates the message $a_i$. For $1\leq i\leq (q-1)$ and $1\leq j\leq n$ a source $s_j \in S_{b_i}$ generates the message $b_{ij}$. And a source $s_i\in S_c$ generates the message $c_i$. There are $2q$ sets of terminals namely, $T_c,T_a$, $T_{b_i}$ and $T_{c_i}$ for $1\leq i\leq (q-1)$; and each of these sets contains $n$ terminals. Each individual terminal is indicated by the source message it demands.

Below we list the set of edges which has a source node as its tail.
\begin{enumerate}
\item $(s,u_1)$ for $\forall s \in \{S_a,S_{b_1},S_{b_2},\ldots,S_{b_{q-1}} \} $.
\item $(s,u_2)$ for $\forall s \in \{S_{b_1},S_{b_2},\ldots,S_{b_{q-1}},S_c\} $.
\item $(a_i,u_{11})$ for $1\leq i\leq n$.
\item $(c_i,u_6)$ for $1\leq i\leq n$.
\item $(b_{ij},tail(e_k))$ for $1\leq i,k \leq (q-1)$, $i\neq k$, and $1\leq j\leq n$.
\item $(b_{ij},v_k)$ for $1\leq i,k \leq (q-1)$, $i\neq k$, and $1\leq j\leq n$.
\item $(b_{ij},w_i)$ for $1\leq i \leq (q-1)$, and $1\leq j\leq n$.
\end{enumerate}
We now list the edges which originates at an intermediate node and ends at a intermediate node.
\begin{enumerate}
\setcounter{enumi}{7}
\item $(u_i,u_{i+2})$ for $1\leq i\leq 7, i\neq 4$.
\item $(u_i,u_{i+1})$ for $i=4,8,9,11,13$.
\item $(u_3,u_6)$, $(u_7,u_{11})$, and $(u_8,u_{13})$
\item $e_i$ for $1\leq i\leq (q-1)$
\item $(u_4,tail(e_i))$ for $1\leq i\leq (q-1)$
\item $(head(e_i),u_{13})$ and $(head(e_i),w_i)$ for $1\leq i\leq (q-1)$
\item $(u_{10},v_i)$ and $(v_i,v^\prime_i)$ for $1\leq i\leq (q-1)$
\item $(w_i,w^\prime_i)$ for $1\leq i\leq (q-1)$
\end{enumerate}
For any terminal $t_i\in T_c$ there exists an edge $(u_{12},t_i)$ and $t_i$ demands the message $c_i$. For any terminal $t_j\in T_{b_i}$ for $1\leq i\leq (q-1), 1\leq j\leq n$, there exits an edge $(v^\prime_i,t_j)$ where the terminal $t_j$ demands the message $b_{ij}$. For any terminal $t_i\in T_a$ there exits an edge $(u_{14},t_i)$ and $t_i$ demands the message $a_i$. For $1\leq i\leq (q-1)$, a terminal $t_j\in T_{c_i}$ for $1\leq j\leq n$ is connected to the node $w^\prime_i$ by the edge $(w^\prime_i,t_j)$ and $t_j$ demands the message $c_j$. The local coding matrices are shown alongside the edges.
\begin{lemma}\label{lem3}
The network in Fig.~\ref{genfano1/n} has a rate $\frac{1}{n}$ fractional linear network coding solution if and only if the characteristic of the finite field divides $q$.
\end{lemma}
\begin{IEEEproof}
Consider a $(d,dn)$ fractional linear network coding solution of the network $\mathcal{N}_1$ where $d$ is any non-zero positive integer. Then the sizes of the local coding matrices are given in the following. For $1\leq i\leq n$, the matrices $D_{4i}$ and $D_{1i}$ are of size $dn\times d$ and it left multiplies the information $a_i$ which is a $d$ length vector. Matrices $A_{ij}$, $B_{ij}$, $U_{ijk}$, $J_{ijk}$ and $E_{ij}$ for $1\leq i,k\leq (q-1),i\neq k$ and $1\leq j\leq n$ are of size $dn\times d$ and it left multiplies the $d$ length vector $b_{ij}$. For $1\leq i\leq n$, the matrices $D_{2i}$ and $D_{3i}$ are of size $dn\times d$ and it left multiplies the information $c_i$. The following matrices are of size $dn\times dn$: $M_i$ for $1\leq i\leq 5$, $Q_1,Q_2$, $K_i,R_i,V_i$ and $W_i$ for $1\leq i\leq (q-1)$. And the following are the matrices of size $d\times dn$: $G_j, R_{ij}, L_j$ and $Z_{ij}$ for $1\leq i\leq (q-1)$ and $1\leq j\leq n$. ALso let $I$ be a $d\times d$ identity matrix. Then,
\begin{IEEEeqnarray}{rl}
Y_{(u_1,u_3)} &= \sum_{i=1}^n D_{1i}a_i + \sum_{i=1}^{q-1} \sum_{j=1}^n A_{ij}b_{ij}\\
Y_{(u_2,u_4)} &= \sum_{i=1}^{q-1} \sum_{j=1}^n B_{ij}b_{ij} + \sum_{i=1}^n D_{2i}c_{i}\\
Y_{(u_5,u_7)} &= M_1Y_{(u_1,u_3)} + M_2Y_{(u_2,u_4)} = \sum_{i=1}^n M_1D_{1i}a_i \IEEEnonumber\\&+\> \sum_{i=1}^{q-1} \sum_{j=1}^n (M_1A_{ij} + M_2B_{ij}) b_{ij} + \sum_{i=1}^n M_2D_{2i}c_{i}\IEEEeqnarraynumspace\\
Y_{(u_6,u_8)} &= M_3Y_{(u_1,u_3)} + \sum_{i=1}^n D_{3i}c_{i} \IEEEnonumber\\&= \sum_{i=1}^n M_3D_{1i}a_i {+} \sum_{i=1}^{q-1} \sum_{j=1}^n M_3A_{ij}b_{ij} + \sum_{i=1}^n D_{3i}c_{i}\IEEEeqnarraynumspace\\
Y_{(u_9,u_{10})} &= M_4Y_{(u_5,u_7)} + M_5Y_{(u_6,u_8)} \IEEEnonumber\\&=\> \sum_{i=1}^n (M_4M_1D_{1i} + M_5M_3D_{1i})a_i \IEEEnonumber\\&+\> \sum_{i=1}^{q-1} \sum_{j=1}^n \{ M_4(M_1A_{ij} + M_2B_{ij}) + M_5M_3A_{ij}\} b_{ij} \IEEEnonumber\\& +\>\sum_{i=1}^n (M_4M_2D_{2i} + M_5D_{3i})c_{i}\IEEEeqnarraynumspace\\
Y_{(u_{11},u_{12})} &= \sum_{i=1}^n (D_{4i} + Q_1M_1D_{1i})a_i \IEEEnonumber \\&+\> \sum_{i=1}^{q-1} \sum_{j=1}^n (M_1A_{ij} + M_2B_{ij}) b_{ij} + \sum_{i=1}^n M_2D_{2i}c_{i}\IEEEeqnarraynumspace
\end{IEEEeqnarray}
\vspace*{-11pt}
\begin{IEEEeqnarray}{rl}
\IEEEeqnarraymulticol{2}{l}{\text{for } 1\leq i\leq (q-1):}\IEEEnonumber\\
 Y_{e_i} &= W_iY_{(u_2,u_4)} + \sum_{j=1,j\neq i}^{q-1} \,\sum_{k=1}^n U_{jki}b_{jk} = \sum_{k=1}^n W_iB_{ik}b_{ik} \IEEEnonumber\\&+\> \sum_{j=1,j\neq i}^{q-1}\, \sum_{k=1}^n (W_iB_{jk} + U_{jki})b_{jk} + \sum_{k=1}^n W_iD_{2k}c_{k}\IEEEeqnarraynumspace\\[12pt]
\IEEEeqnarraymulticol{2}{l}{\text{for } 1\leq i\leq (q-1):}\IEEEnonumber\\
Y_{(v_i,v_i^\prime)} &= K_iY_{(u_9,u_{10})} + \sum_{j=1,j\neq i}^{q-1} \,\sum_{k=1}^n J_{jki}b_{jk} \IEEEnonumber\\&=\> \sum_{k=1}^n K_i(M_4M_1D_{1k} + M_5M_3D_{1k})a_k \IEEEnonumber\\&+\>  \sum_{j=1}^n K_i\{ M_4(M_1A_{ij} + M_2B_{ij}) + M_5M_3A_{ij}\} b_{ij} \IEEEnonumber\\&+\> \sum_{k=1,k\neq i}^{q-1} \sum_{j=1}^n \{J_{kji} + K_i(M_4(M_1A_{kj} + M_2B_{kj})\IEEEnonumber \\&\hfill+ M_5M_3A_{kj})\} b_{kj} \IEEEnonumber\\&{+}\!\! \sum_{j=1}^n\! K_i(M_4M_2D_{2j} {+} M_5D_{3j})c_{j}\IEEEeqnarraynumspace\\[12pt]
\IEEEeqnarraymulticol{2}{l}{\text{for } 1\leq i\leq (q-1):}\IEEEnonumber\\
Y_{(w_i,w_i^\prime)} &= V_iY_{e_i} + \sum_{j=1}^n E_{ij}b_{ij} = \sum_{k=1}^n (V_iW_iB_{ik} + E_{ik}) b_{ik} \IEEEnonumber \\&+\> \sum_{j=1,j\neq i}^{q-1}\, \sum_{k=1}^n \{ V_i(W_iB_{jk} + U_{jki})\} b_{jk} \IEEEnonumber \\&+\> \sum_{k=1}^n V_iW_iD_{2k}c_{k}\IEEEeqnarraynumspace\\
\IEEEeqnarraymulticol{2}{l}{Y_{(u_{13},u_{14})} = Q_2Y_{(u_6,u_8)} + \sum_{i=1}^{q-1} R_iY_{e_i} = \sum_{i=1}^n Q_2M_3D_{1i}a_i}\IEEEnonumber\\
&\;\;\;\,+\> \sum_{i=1}^{q-1} \sum_{j=1}^n \{Q_2M_3A_{ij} + R_iW_iB_{ij} \IEEEnonumber\\
&\;\;\;\,\hfill +\> \sum_{k=1,k\neq i}^{q-1} R_k(W_kB_{ij}+U_{ijk})\}b_{ij}\IEEEnonumber\\
&\;\;\;\,+\> \sum_{i=1}^n\{ Q_2D_{3i} + \big(\sum_{k=1}^{q-1}R_kW_k\big) D_{2i}\} c_i
\end{IEEEeqnarray}

Since a terminal $t_i\in T_c$ computes $c_i$, we have the following inequalities. For $1\leq i,j\leq n$ since the component of $a_i$ is zero at all $t_j\in T_c$,
\begin{equation}
G_j(Q_1M_1D_{1i} + D_{4i}) = 0
\end{equation}
As the components of for $b_{ij}$ is also zero at all $t_k\in T_c$, for $1\leq i\leq (q-1)$ and $1\leq j,k\leq n$ we have,
\begin{equation}
G_k\{Q_1(M_1A_{ij} + M_2B_{ij})\} = 0\label{etcba}
\end{equation}
Now, since the terminal $t_i\in T_c$ retrieves $c_i$ for $1\leq i,j\leq n$ and $j\neq i$ we have,
\begin{IEEEeqnarray}{l}
G_i(Q_1M_2D_{2i}) = I\label{etc3a}\\
G_i(Q_1M_2D_{2j}) = 0\label{etc3b}
\end{IEEEeqnarray}

Now consider the $n$ terminals in the set $T_{b_i}$ for $1\leq i\leq (q-1)$. Since the component of $a_k$ for $1\leq k\leq n$ at $t_j\in T_{b_i}$ for $1\leq j\leq n$ is zero, we have, for $1\leq i\leq (q-1)$ and $1\leq j,k\leq n$:
\begin{equation}
R_{ij}K_i(M_4M_1D_{1k} + M_5M_3D_{1k}) = 0 \label{etba1}
\end{equation}
Since the terminal $t_j\in T_{b_i}$ computes the information $b_{ij}$, we have, for $1\leq i,k\leq (q-1), i\neq k$, $1\leq j,m,l\leq n$ and $m\neq j$:
\begin{IEEEeqnarray}{l}
R_{ij}K_i\{ M_4(M_1A_{ij} + M_2B_{ij}) + M_5M_3A_{ij}\} = I\label{etbb1}\\
R_{ij}K_i\{ M_4(M_1A_{im} + M_2B_{im}) + M_5M_3A_{im}\} = 0\label{etbb2}\\
R_{ij}\{K_i( M_4(M_1A_{kl} + M_2B_{kl}) + M_5M_3A_{kl}) {+} J_{kli}\} = 0\IEEEeqnarraynumspace
\end{IEEEeqnarray}
Since the component of $c_k$ for $1\leq k\leq n$ is zero at $t_j\in T_{b_i}$, we have for $1\leq i\leq (q-1)$ and $1\leq j,k\leq n$,
\begin{equation}
R_{ij}K_i(M_4M_2D_{2k} + M_5D_{3k}) = 0 \label{etbc1}
\end{equation}

Let us consider the terminals in the set $T_{a}$. Since $t_i\in T_a$ computes the message $a_i$, for $1\leq i,j\leq n$ and $j\neq i$, we have:
\begin{IEEEeqnarray}{l}
L_iQ_2M_3D_{1i} = I \label{etaa1}\\
L_iQ_2M_3D_{1j} = 0 \label{etaa2}
\end{IEEEeqnarray}
At any $t_{l}\in T_a$ for $1\leq i\leq (q-1)$ and $1\leq l,j\leq n$ the component of $b_{ij}$ is zero. So we have,
\begin{equation}
L_l\{Q_2M_3A_{ij} + R_iW_iB_{ij} + \sum_{k=1,k\neq i}^{q-1} R_k(W_kB_{ij}+U_{ijk})\}\label{etab1}
\end{equation}
At a terminal $t_j\in T_a$, since the component of $c_i$ is zero, for $1\leq i,j\leq n$:
\begin{equation}
L_j\{ Q_2D_{3i} + \big(\sum_{k=1}^{q-1}R_kW_k\big) D_{2i}\} = 0 \label{etac1}
\end{equation}

Now consider the terminals in the set $T_{c_i}$ for $1{\leq} i{\leq} (q{-}1)$. Since at $t_{l}\in T_{c_i}$ the component of $b_{ij}$ for $1\leq i\leq (q-1)$ and $1\leq l,j\leq n$ is zero, we have,
\begin{equation}
Z_{il}(V_iW_iB_{ij} + E_{ij}) = 0
\end{equation}
For $1\leq k\leq (q-1),k\neq i$ the component of $b_{kj}$ for $1\leq j\leq n$ is zero as well,
\begin{equation}
Z_{il}\{V_i(W_iB_{kj} + U_{kji})\} = 0 \label{etci1}
\end{equation}
Since $t_l\in T_{c_i}$ computes $c_l$, for $1\leq l,m\leq n, l\neq m$, we have,
\begin{IEEEeqnarray}{l}
Z_{il}V_iW_iD_{2l} = I \label{etcc1}\\
Z_{il}V_iW_iD_{2m} = 0 \label{etcc2}
\end{IEEEeqnarray}

The rest of the proof requires a lemma and a corollary which are presented next.

Let us consider $A = {\begin{bmatrix} A_1 & A_2 & \cdots & A_n  \end{bmatrix}}^T$ and $B = \begin{bmatrix} B_1 & B_2 &\cdots & B_n \end{bmatrix}$ where $A_i$ and $B_i$ for $1\leq i\leq n$ are matrices of size $d\times dn$ and $dn\times d$ respectively. Let $I_{dn}$ denote an identity matrix of size $dn \times dn$.
\begin{lemma}\label{lem1}
If $A_iB_i = I$ but $A_iB_j = 0$ for $1\leq i,j\leq n, i\neq j$, then $AB = I_{dn}$; and both $A$ and $B$ has an unique inverse.
\end{lemma}
\begin{IEEEproof}
\begin{IEEEeqnarray*}{rl}
AB &= \begin{bmatrix} A_1 \\ A_2 \\ \vdots \\ A_n  \end{bmatrix} \begin{bmatrix} B_1 & B_2 & \cdots & B_n  \end{bmatrix}\\
&= \begin{bmatrix}
A_1B_1 & A_1B_2 & \cdots & A_1B_n\\
A_2B_1 & A_2B_2 & \cdots & A_2B_n\\
\vdots &\vdots &\vdots &\vdots\\
A_nB_1 & A_nB_2 & \cdots & A_nB_n
\end{bmatrix}\\
&= \begin{bmatrix}
I & 0 & \cdots & 0\\
0 & I & \cdots & 0\\
\vdots &\vdots &\vdots &\vdots\\
0 & 0 & \cdots & I
\end{bmatrix} = I_{dn}
\end{IEEEeqnarray*}
Now since both $A$ and $B$ are matrices of size $dn\times dn$, for $AB = I_{dn}$ to hold, both of $A$ and $B$ has to be full rank matrices. This completes the proof.
\end{IEEEproof}
\begin{corollary}\label{coro1}
For $1\leq i,j\leq n$, if $A_iB_j = 0$, then $AB = 0$.
\end{corollary}

We now define the following matrices. Note that all of these newly defined matrices are square matrices of size $dn \times dn$.
\begin{IEEEeqnarray}{l}
G = {\begin{bmatrix} G_1 & G_2 & \cdots & G_n\end{bmatrix}}^T \label{g}\\
Q_1(M_1A_{i} + M_2B_{i}) = \begin{bmatrix}Q_1(M_1A_{i1} + M_2B_{i1}) \\ Q_1(M_1A_{i2} + M_2B_{i2}) \\ \vdots \\ Q_1(M_1A_{in} + M_2B_{in})  \end{bmatrix}^T\label{q}\\
Q_1M_2D_2 = \begin{bmatrix} Q_1M_2D_{21} & Q_1M_2D_{22} \!& {\cdots} &\! Q_1M_2D_{2n} \end{bmatrix}\IEEEeqnarraynumspace\label{q1}\\
R_iK_i = \begin{bmatrix}R_{i1}K_i & R_{i2}K_i & \cdots & R_{in}K_i \end{bmatrix}^T\label{r}
\end{IEEEeqnarray}
\begin{IEEEeqnarray}{l}
M_4M_1D_{1} + M_5M_3D_{1} = \begin{bmatrix}M_4M_1D_{11} + M_5M_3D_{11} \\ M_4M_1D_{12} + M_5M_3D_{12} \\ \vdots \\ M_4M_1D_{1n} + M_5M_3D_{1n} \end{bmatrix}^T\label{m4}\\
M_4(M_1A_{i} + M_2B_{i}) + M_5M_3A_{i} \IEEEnonumber \\\hfill =\> \begin{bmatrix} M_4(M_1A_{i1} + M_2B_{i1}) + M_5M_3A_{i1} \\ M_4(M_1A_{i2} + M_2B_{i2}) + M_5M_3A_{i2} \\ \vdots \\ M_4(M_1A_{in} + M_2B_{in}) + M_5M_3A_{in} \end{bmatrix}^T\IEEEeqnarraynumspace\label{m4m2m5}\\
M_4M_2D_{2} + M_5D_{3} = \begin{bmatrix} M_4M_2D_{21} + M_5D_{31} \\ M_4M_2D_{22} + M_5D_{32} \\ \vdots \\ M_4M_2D_{2n} + M_5D_{3n} \end{bmatrix}^T\label{m4m5}\\
L = \begin{bmatrix}L_1 & L_2 & \cdots & L_n \end{bmatrix}^T\label{l}\\
Q_2M_3D_1 {=} \begin{bmatrix}Q_2M_3D_{11} & Q_2M_3D_{12} \!&{\cdots} &\! Q_2M_3D_{1n} \end{bmatrix}\label{q2} \IEEEeqnarraynumspace
\end{IEEEeqnarray}\vspace*{-12pt}
\begin{IEEEeqnarray}{l}
Q_2M_3A_{i} + R_iW_iB_{i} + \sum_{k=1,k\neq i}^{q-1} R_k(W_kB_{i}{+}U_{ik})\IEEEnonumber\\
{=}\!\! \begin{bmatrix}\! Q_2M_3A_{i1} {+} R_iW_iB_{i1} {+} \sum_{k{=}1,k{\neq} i}^{q-1} R_k(W_kB_{i1}{+}U_{i1k})\! \\
\! Q_2M_3A_{i2} {+} R_iW_iB_{i2} {+} \sum_{k{=}1,k{\neq} i}^{q-1} R_k(W_kB_{i2}{+}U_{i2k})\!\\ \vdots \\
\! Q_2M_3A_{in} {+} R_iW_iB_{in} {+} \sum_{k{=}1,k{\neq} i}^{q-1} R_k(W_kB_{in}{+}U_{ink}) \! \end{bmatrix}^T \IEEEeqnarraynumspace\!\!\!\!\!\!\label{q2r}\\
Q_2D_{3} + \big(\sum_{k=1}^{q-1}R_kW_k\big) D_{2} \IEEEnonumber\\ \hspace{80pt}=\> \begin{bmatrix} Q_2D_{31} + \big(\sum_{k=1}^{q-1}R_kW_k\big) D_{21} \\ Q_2D_{32} + \big(\sum_{k=1}^{q-1}R_kW_k\big) D_{22}  \\ \vdots \\ Q_2D_{3n} + \big(\sum_{k=1}^{q-1}R_kW_k\big) D_{2n}\end{bmatrix}^T\label{q2d2}\\
Z_{i}V_i = \begin{bmatrix} Z_{i1}V_i & Z_{i2}V_i & \cdots & Z_{in}V_i \end{bmatrix}^T\label{z}\\
W_iB_{k} + U_{ki} = \begin{bmatrix} V_i(W_iB_{k1} + U_{k1i}) \\ V_i(W_iB_{k2} + U_{k2i}) \\ \vdots \\ V_i(W_iB_{kn} + U_{kni}) \end{bmatrix}^T\label{u}\\
W_iD_{2} = \begin{bmatrix} W_iD_{21} & W_iD_{22} & \cdots & W_iD_{2n} \end{bmatrix}\label{wd2}
\end{IEEEeqnarray}
For the rest of this section we redefine $I$ as an identity matrix of size $dn \times dn$.
Using the Corollary~\ref{coro1}, from equations (\ref{etcba}), (\ref{g}) and (\ref{q}), we have, for $1\leq i\leq (q-1)$:
\begin{equation}
G\{Q_1(M_1A_{i} + M_2B_{i})\} = 0 \label{etcb2}
\end{equation}
Using lemma~\ref{lem1}, from equations (\ref{etc3a}), (\ref{etc3b}), (\ref{g}) and (\ref{q1}) we get,
\begin{equation}
G(Q_1M_2D_2) = I \label{etc3c}
\end{equation}
Using Corollary~\ref{coro1} and equations (\ref{etba1}), (\ref{r}) and (\ref{m4}) we have, for $1\leq i\leq (q-1)$:
\begin{equation}
R_iK_i(M_4M_1D_{1} + M_5M_3D_{1}) = 0 \label{etba2}
\end{equation}
Using Lemma~\ref{lem1} and equations (\ref{etbb1}), (\ref{etbb2}), (\ref{r}) and (\ref{m4m2m5})  we have, for $1\leq i\leq (q-1)$:
\begin{equation}
R_{i}K_i\{ M_4(M_1A_{i} + M_2B_{i}) + M_5M_3A_{i}\} = I \label{etbb3}
\end{equation}
Using Corollary~\ref{coro1} on equations (\ref{etbc1}), (\ref{r}) and (\ref{m4m5}) we have, for $1\leq i\leq (q-1)$:
\begin{equation}
R_{i}K_i(M_4M_2D_{2} + M_5D_{3}) = 0 \label{etbc2}
\end{equation}
Using Lemma~\ref{lem1} on equations (\ref{etaa1}), (\ref{etaa2}), (\ref{l}) and (\ref{q2}) for we get:
\begin{equation}
LQ_2M_3D_1 = I \label{etaa3}
\end{equation}
Using Corollary~\ref{coro1} and equations (\ref{etab1}), (\ref{l}) and (\ref{q2r}) we get, for $1\leq i\leq (q-1)$:
\begin{equation}
L(Q_2M_3A_{i} + R_iW_iB_{i} + \sum_{k=1,k\neq i}^{q-1} R_k(W_kB_{i} + U_{ik})) = 0 \label{etab2}
\end{equation}
Using Lemma~\ref{lem1}, and equations (\ref{etac1}), (\ref{l}) and (\ref{q2d2}) we have:
\begin{equation}
L(Q_2D_{3} + \big(\sum_{k=1}^{q-1}R_kW_k\big) D_{2}) = 0 \label{etac2}
\end{equation}
Using Corollary~\ref{coro1} and equations (\ref{etci1}), (\ref{z}) and (\ref{u}) we get, for $1\leq i,k\leq (q-1), k\neq i$,
\begin{equation}
Z_{i}V_i(W_iB_{k} + U_{ki}) = 0 \label{etci2}
\end{equation}
Using Lemma~\ref{lem1}, and equations (\ref{etcc1}), (\ref{etcc2}), (\ref{z}) and (\ref{wd2}) we have:
\begin{equation}
Z_{i}V_iW_iD_{2} = I \label{etcc3}
\end{equation}

Now, from equation (\ref{etc3c}) the matrices $G$, $D_2$, $M_2$  and $Q_1$ are invertible. From equation (\ref{etbb3}), for $1\leq i\leq (q-1)$, $K_i$ and $R_i$ are invertible. $D_1$, $L$ and $Q_2$ is invertible from equation (\ref{etaa3}). From equation (\ref{etcc3}) $V_i$ and $Z_i$ are invertible. From (\ref{etcb2}) since $G$ and $Q_1$ both are invertible, we have:
\begin{equation}
M_1A_{i} + M_2B_{i} = 0 \label{etcb3}
\end{equation}
Since bot $R_i$ and $K_i$ are invertible, we have from equation (\ref{etba2}):
\begin{equation}
M_4M_1D_{1} + M_5M_3D_{1} = 0\label{neweq1}
\end{equation}
And from (\ref{etbc2}) we have:
\begin{equation}
M_4M_2D_{2} + M_5D_{3} = 0 \label{etbc3}
\end{equation}
Since $D_1$ is also invertible, from equation (\ref{neweq1}) we must have,
\begin{equation}
M_4M_1 + M_5M_3 = 0 \label{etba3}
\end{equation}
Substituting (\ref{etcb3}) in (\ref{etbb3}) we have
\begin{equation}
R_{i}K_iM_5M_3A_{i} = I \label{etbb4}
\end{equation}
Since $Z_i$ and $V_i$ are invertible, from equation (\ref{etci2}) we must have, for $1\leq i,k\leq (q-1), i\leq k$:
\begin{equation}
W_iB_{k} + U_{ki} = 0 \label{etci3}
\end{equation}
Substituting equation (\ref{etci3}) in equation (\ref{etab2}), and noting that $L$ is invertible, we have for $1\leq i\leq (q-1)$:
\begin{equation}
Q_2M_3A_{i} + R_iW_iB_{i} = 0 \label{etab3}
\end{equation}
$M_3$ and $A_i$ are invertible from (\ref{etbb4}) for $1\leq i\leq (q-1)$, and since $Q_2$ is invertible, $Q_2M_3A_i$ is invertible, which leads that $R_iW_iB_{i}$ is invertible from equation (\ref{etab3}), and hence $B_i$ is invertible. So from (\ref{etab3}), for $1\leq i\leq (q-1)$,
\begin{equation}
Q_2M_3A_{i}B_{i}^{-1} + R_iW_i = 0 \label{etab4}
\end{equation}
Since $L$ is invertible, we have from (\ref{etac2})
\begin{equation}
Q_2D_{3} + \sum_{k=1}^{q-1}R_kW_kD_{2} = 0
\end{equation}
Since $D_2$ is invertible, we have:
\begin{equation}
Q_2D_{3}D_2^{-1} + \sum_{k=1}^{q-1}R_kW_k = 0 \label{etac3}
\end{equation}
Substituting $R_kW_k$ for $1\leq k\leq (q-1)$, from (\ref{etab4}) we have,
\begin{IEEEeqnarray}{l}
Q_2D_{3}D_2^{-1} + \sum_{k=1}^{q-1}-Q_2M_3A_{k}B_{k}^{-1} = 0 \label{e1}
\end{IEEEeqnarray}
Since $M_2, A_i$ and $B_i$ are invertible for $1\leq i\leq (q-1)$, from (\ref{etcb3}) $M_1$ is also invertible, and hence we have, for $1\leq i\leq (q-1)$:
\begin{equation}
A_{i}B^{-1}_i = - M_1^{-1}M_2 \label{etcb4}
\end{equation}
Substituting equation (\ref{etcb4}) in equation (\ref{e1}) we have:
\begin{IEEEeqnarray*}{l}
Q_2D_{3}D_2^{-1} + \sum_{k=1}^{q-1}-Q_2M_3(-M_1^{-1}M_2) = 0
\end{IEEEeqnarray*}
Since all constituent matrices are square,
\begin{IEEEeqnarray}{l}
Q_2D_{3}D_2^{-1} + \sum_{k=1}^{q-1}Q_2M_3M_1^{-1}M_2 = 0 \label{e2}
\end{IEEEeqnarray}
Since $M_5M_3$ is invertible from equation (\ref{etbb4}), and $M_1$ is invertible as shown above, from equation (\ref{etba3}) we have:
\begin{equation}
M_3M_1^{-1} = -M_5^{-1}M_4 \label{etba4}
\end{equation}
Substituting equation (\ref{etba4}) in equation (\ref{e2}) we have:
\begin{IEEEeqnarray}{l}
Q_2D_{3}D_2^{-1} + \sum_{k=1}^{q-1}Q_2(-M_5^{-1}M_4)M_2 = 0 \label{e3}
\end{IEEEeqnarray}
Now, from equation (\ref{etbc3}) we have:
\begin{equation}
D_{3}D_2^{-1} = -M_5^{-1}M_4M_2
\end{equation}
Substituting this in (\ref{e3}) we have:
\begin{IEEEeqnarray*}{l}
Q_2D_{3}D_2^{-1} + \sum_{k=1}^{q-1}Q_2D_{3}D_2^{-1} = 0\\
(q)Q_2D_{3}D_2^{-1} = 0
\end{IEEEeqnarray*}
Since $Q_2$,$D_3$ and $D_2$ are all invertible matrices, it must be that $q=0$. Now the fact that an element is equal to zero in a finite field if and only if the characteristic divides the element proves that the network has a rate $\frac{1}{n}$ fractional linear network coding solution only if the characteristic of the finite field divides $q$.

We now show that the network $\mathcal{N}_1$ has a $(1,n)$ fractional linear network coding solution if $q=0$. For this section, let $\bar{a}_i$ denote an $n$-length vector whose $i^{\text{th}}$ component is $a_i$ and all other component is zero. Let $\bar{c}_i$ to denote an $n$-length vector whose $i^{\text{th}}$ component is $c_i$ and all other component is zero. Also let $\bar{b}_{ij}$ denotes an $n$-length vector which has zero in all of its components but the $j^{\text{th}}$ one, which is equal to $b_{ij}$. Note that $a_i$, $c_i$ for $1\leq i\leq n$ and $b_{ij}$ for $1\leq i\leq (q-1), 1\leq j\leq n$ are the source processes. Now, it can be seen that by choosing the appropriate local coding matrices, the messages shown below can be transmitted by the corresponding edges.
\begin{IEEEeqnarray*}{l}
Y_{(u_1,u_3)} = \sum_{i=1}^{n} \bar{a}_i + \sum_{i=1}^{p-1} \sum_{j=1}^{n} \bar{b}_{ij}\\
Y_{(u_2,u_4)} = \sum_{i=1}^{q-1} \sum_{j=1}^{n} \bar{b}_{ij} + \sum_{i=1}^{n} \bar{c}_i\\
Y_{(u_5,u_7)} = Y_{(u_1,u_3)} - Y_{(u_2,u_4)} = \sum_{i=1}^{n} \bar{a}_i - \sum_{i=1}^{n} \bar{c}_i\\
Y_{(u_6,u_8)} = Y_{(u_1,u_3)} - \sum_{i=1}^{n} \bar{c}_i = \sum_{i=1}^{n} \bar{a}_i + \sum_{i=1}^{q-1} \sum_{j=1}^{n} \bar{b}_{ij} - \sum_{i=1}^{n} \bar{c}_i\\
\text{for } 1\leq i\leq q-1: \quad Y_{e_i} = \sum_{j=1}^{n} \bar{b}_{ij} + \sum_{i=1}^{n} \bar{c}_i\\
Y_{(u_9,u_{10})} = Y_{(u_6,u_8)} - Y_{(u_5,u_7)} = \sum_{i=1}^{q-1} \sum_{j=1}^{n} \bar{b}_{ij}\\
Y_{(u_{13},u_{14})} = Y_{(u_6,u_8)} - \sum_{i=1}^{q-1} Y_{e_i} = \sum_{i=1}^{n} \bar{a}_i - \sum_{i=1}^{n} \bar{c}_i - \sum_{i=1}^{q-1} \sum_{i=1}^{n} \bar{c}_i \\=\> \sum_{i=1}^{n} \bar{a}_i - \sum_{i=1}^{q} \sum_{i=1}^{n} \bar{c}_i = \sum_{i=1}^{n} \bar{a}_i -  \sum_{i=1}^{n} q\bar{c}_i = \sum_{i=1}^{n} \bar{a}_i\\
Y_{(u_{11},u_{12})} = \sum_{i=1}^{n} \bar{a}_i - Y_{(u_5,u_7)} = \sum_{i=1}^{n} \bar{c}_i\\
\text{for } 1\leq i\leq q-1:\\ Y_{(v_i,v_i^\prime)} = Y_{(u_9,u_{10})} - \sum_{k=1,k\neq i}^{q-1} \sum_{j=1}^{n} \bar{b}_{ij} = \sum_{j=1}^{n} \bar{b}_{ij}\\
\text{for } 1\leq i\leq q-1: \quad Y_{(w_i,w_i^\prime)} = Y_{e_i} - \sum_{j=1}^{n} \bar{b}_{ij} = \sum_{i=1}^{n} \bar{c}_i
\end{IEEEeqnarray*}
Let $\check{u}(i)$ be a unit row vector of length $n$ which has $i^{\text{th}}$ component equal to one and all other component equal to zero. Then from the vector $\sum_{i=1}^{n} \bar{a}_i$, $a_i$ for any $1\leq i\leq n$ can be determined by the dot product $\check{u}(i) \cdot (\sum_{i=1}^{n} \bar{a}_i)$. Similarly for any $1\leq i\leq (q-1)$, $b_{ij} = \check{u}(j) \cdot (\sum_{j=1}^{n} \bar{b}_{ij})$. For $1\leq i\leq n$, $c_i$ can be determined similarly from $\sum_{i=1}^{n} \bar{c}_i$.
\end{IEEEproof}

\begin{theorem}\label{thm1}
For any non-zero positive rational number $\frac{k}{n}$ and for any finite set of prime numbers $\{p_1,p_2,\ldots,p_l\}$ there exists a network which has a rate $\frac{k}{n}$ fractional linear network coding solution if and only if the characteristic of the finite field belong to the given set of primes.
\end{theorem}
\begin{IEEEproof}
Let us consider the union of $k$ copies of the network $\mathcal{N}_1$ shown in Fig.~\ref{genfano1/n} each for $q = p_1\times p_2\times \cdots\times p_l$. Denote the $i^{\text{th}}$ copy as $\mathcal{N}_{1i}$. Note that each source and each terminal has $k$ copies in the union. Join all copies of any source or terminal into a single source or terminal respectively. Name this new network as $\mathcal{N}_1^\prime$. We show below that $\mathcal{N}_1^\prime$ has a rate $\frac{k}{n}$ fractional linear network coding solution if and only if the characteristic of the finite field belong to the set $\{p_1,p_2,\ldots,p_l\}$. Before we proceed further, consider the following property of $\mathcal{N}_1$ and $\mathcal{N}_1^\prime$.
\begin{lemma}\label{lem2}
If $\mathcal{N}_1^\prime$ has a $(dk,dn)$ fractional linear network coding solution for any non-zero positive integer $d$, then $\mathcal{N}_1$ has a $(dk,dkn)$ fractional linear network coding solution.
\end{lemma}
\begin{proof}
This is true since the information that can be sent using the network $\mathcal{N}_1^\prime$ for $x$ times, can be sent using the network $\mathcal{N}_1$ $kx$ times. This is because $\mathcal{N}_1^\prime$ has $k$ copies of $\mathcal{N}_1$.
\end{proof}
First consider the only if part. Say $\mathcal{N}_1^\prime$ has a rate $\frac{k}{n}$ fractional linear network coding solution even if the characteristic does not belong to the set $\{p_1,p_2,\ldots,p_l\}$. Then from Lemma~\ref{lem2}, the network $\mathcal{N}_1$ has a rate $\frac{1}{n}$ fractional linear network coding solution even if the characteristic does not belong to the given set of primes. However, as shown in Lemma~\ref{lem3}, the $\mathcal{N}_1$ has a rate $\frac{1}{n}$ fractional linear network coding solution if and only if $q=0$ over the finite field. But, as $q = p_1\times p_2\times \cdots\times p_l$, $q=0$ if and only if one of the prime number from the set $\{p_1,p_2,\ldots,p_l\}$ is zero over the finite field. The latter is the case if and only if the characteristic of the finite field is one of the primes in the set. Hence this is a contradiction to the case that $\mathcal{N}_1$ has a rate $\frac{1}{n}$ fractional linear network coding solution even if the characteristic does not belong to the given set of primes.

Now consider the if part. Since $\mathcal{N}_{1i}$ for $1\leq i\leq k$ has a $(1,n)$ fractional linear network coding solution, a $(k,n)$ fractional linear network coding solution can be constructed by keeping the same local coding matrices in all of the copies and sending the $i^{\text{th}}$ component of each source through $\mathcal{N}_{1i}$.
\end{IEEEproof}

\subsection{Network having $\frac{k}{n}$ solution iff characteristic belongs to a given co-finite set of primes.}\label{subsec2}
\begin{figure*}
\centering
\includegraphics[width=\textwidth]{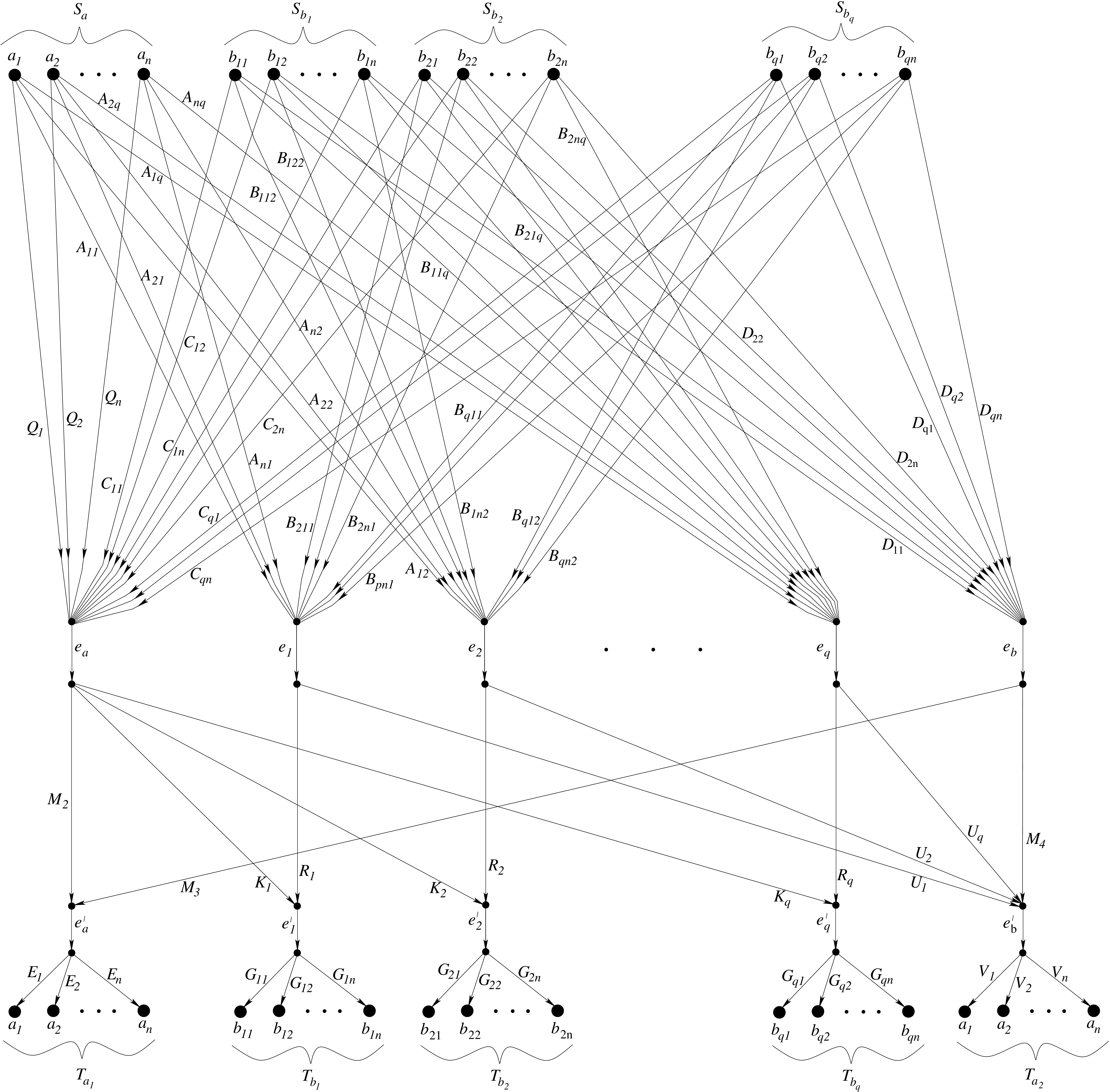}
\caption{A network $\mathcal{N}_2$ which has a rate $1/n$ fractional linear network coding solution if and only if the characteristic of the finite field does not divide $q$.}
\label{nofano1/n}
\end{figure*}
The outline of the contents in this section is similar to that of the last sub-section. Consider the network $\mathcal{N}_2$ shown in Fig.~\ref{nofano1/n}. The sources are partitioned into $q+1$ sets: $S_a$ and $S_{b_i}$ for $1\leq i\leq q$. Each of these sets has $n$ sources. In Fig.~\ref{nofano1/n} the sources are indicated by the message symbol it generates. A source $s_i\in S_a$ generates the message $a_i$. Similarly a source $s_j \in S_{b_i}$ generates the message $b_{ij}$ for $1\leq i\leq q, 1\leq j\leq n$. The set of terminals are partitioned into $q+2$ disjoint sets namely, $T_{a_1},T_{a_2}$ and $T_{b_i}$ for $1\leq i\leq q$. The sets $T_{a_1}$ and $T_{a_2}$ has $n$ terminals in each. Each individual terminal is indicated by the source message it demands. We have the following edges in the network.
\begin{enumerate}
\item $e_a,e_b,e_a^\prime$ and $e_b^\prime$
\item $e_i$ and $e_i^\prime$ for $1\leq i\leq q$
\item $(s,tail(e_a))$ for $\forall s\in  S_a \cup \{\cup_{i=1}^q S_{b_i}\}$
\item $(s,tail(e_i))$ for $1\leq i\leq q$ and $\forall s\in \{ S_a, \cup_{j=1,j\neq i}^q S_{b_i}\}$
\item $(s,tail(e_b))$ for $\forall s\in \cup_{i=1}^q S_{b_i}$
\item $(head(e_a),tail(e_a^\prime))$ and $(head(e_b),tail(e_a^\prime))$
\item $(head(e_b),tail(e_b^\prime))$
\item $(head(e_i),tail(e_i^\prime))$ for $1\leq i\leq q$
\item $(head(e_a),tail(e_i^\prime))$ for $1\leq i\leq q$
\item $(head(e_i),tail(e_b^\prime))$ for $1\leq i\leq q$
\end{enumerate}
From each of the nodes $head(e_a^\prime)$, $head(e_i^\prime)$ for $1\leq i\leq q$ and $head(e_b)^\prime$, $n$ outgoing edge emanates and the $head$ node of all such edges is a terminal. The set of $n$ terminals which have a path from node $head(e_a^\prime)$ are denoted by $T_{a_1}$. Similarly, the set of $n$ terminals which have a path from node $head(e_b^\prime)$ are denoted by $T_{a_2}$. And the $n$ terminals in the set $T_{b_i}$ for $1\leq i\leq q$ are connected from the node $head(e_i^\prime)$ by an edge.

Before we show that for any non-zero positive rational number $\frac{k}{n}$ and for a given co-finite set of primes there exist a network which has a rate $\frac{k}{n}$ fractional linear network coding solution if and only if the characteristic of the finite field belong to the given set, consider the following lemma.
\begin{lemma}\label{lem4}
The network shown in Fig.~\ref{nofano1/n} has a rate $\frac{1}{n}$ fractional linear network coding solution if and only if the characteristic of the finite field does not divides $q$.
\end{lemma}
\begin{IEEEproof}
Consider a $(d,dn)$ fractional linear network coding solution of the network. The local coding matrices are shown next to the edges. The matrices $Q_i$ for $1\leq i\leq n$ and $A_{ij}$ for $1\leq i\leq n, 1\leq j\leq q$ are of size $dn \times d$ and left multiplies the massage $a_i$. The matrices $C_{ij}$ for $1\leq i\leq q, 1\leq j\leq n$, $B_{ijk}$ for $1\leq i,k\leq q, i\neq k, 1\leq j\leq n$, and $D_{ij}$ for $1\leq i\leq q, 1\leq j\leq n$ left multiplies $b_{ij}$ and are of size $dn \times d$. The matrices $M_2,M_3,M_4$ and $K_i,R_i$ and $U_i$ for $1\leq i\leq q$ are of sizes $dn \times dn$. And the matrices of size $d \times dn$ are $E_j,G_{ij}$ and $V_j$ for $1\leq i\leq q$ and $1\leq j\leq n$. Also let $I$ be a $d\times d$ identity matrix. The following is the list of messages carried by some of the edges of the network.
\begin{IEEEeqnarray}{ll}
Y_{e_a} &= \sum_{i=1}^n Q_ia_i + \sum_{i=1}^q \sum_{j=1}^n C_{ij}b_{ij}\\
\IEEEeqnarraymulticol{2}{l}{\text{for } 1\leq i\leq q:}\IEEEnonumber\\
Y_{e_i} &= \sum_{j=1}^n A_{ji}a_j + \sum_{j=1,j\neq i}^q \sum_{k=1}^n B_{jki}b_{jk}\\
Y_{e_b} &= \sum_{i=1}^q \sum_{j=1}^n D_{ij}b_{ij}\\
Y_{e_a^\prime} &= M_2Y_{e_a} + M_3Y_{e_b} = \sum_{i=1}^n M_2Q_ia_i \IEEEnonumber\\&\hfill +\> \sum_{i=1}^q \sum_{j=1}^n (M_2C_{ij} + M_3D_{ij})b_{ij}\IEEEeqnarraynumspace\\
\IEEEeqnarraymulticol{2}{l}{\text{for } 1\leq i\leq q:}\IEEEnonumber\\
Y_{e_i^\prime} &= K_iY_{e_a} + R_iY_{e_i} = \sum_{j=1}^n (K_iQ_j + R_iA_{ji})a_j \IEEEnonumber \\&\hfill +\> \sum_{k=1}^n K_iC_{ik}b_{ik} + \sum_{j=1,j\neq i}^q \sum_{k=1}^n (K_iC_{jk} + R_iB_{jki})b_{jk}\IEEEeqnarraynumspace\\
Y_{e_b^\prime} &= \sum_{i=1}^q U_iY_{e_i} + M_4Y_{e_b} =  \sum_{i=1}^q \sum_{j=1}^n U_iA_{ji}a_{j} \IEEEnonumber\\&\hfill +\> \sum_{j=1}^q \sum_{k=1}^n \big( ( \sum_{i=1,i\neq j}^q U_iB_{jki}) + M_4D_{jk} \big)b_{jk} \IEEEeqnarraynumspace
\end{IEEEeqnarray}

Because of the demands of the terminals the following inequalities must be satisfied. Since any terminal $t_i\in T_{a_1}$ computes $a_i$, we have, for $1\leq i,j\leq n, j\neq i$:
\begin{IEEEeqnarray}{l}
E_iM_2Q_i = I\label{gq1}\\
E_iM_2Q_j = 0\label{gq2}
\end{IEEEeqnarray}
At $t_k\in T_{a}$ the component of $b_{ij}$ is zero for $1\leq i\leq q, 1\leq j,k\leq n$ we have:
\begin{equation}
E_k(M_2C_{ij} + M_3D_{ij}) = 0\label{gq3}
\end{equation}
Now consider the terminals in the set $T_{b_i}$ for $1\leq i\leq q$. Since at any terminal $t_j\in T_{b_i}$ for $1\leq j\leq n$ the component of $a_k$ for $1\leq k\leq n$ is zero, we have,
\begin{equation}
G_{ij}(K_iQ_k + R_iA_{ki}) = 0\label{gq4}
\end{equation}
Because $t_j\in T_{b_i}$ computes $b_{ij}$ for $1\leq i\leq q, 1\leq j,k\leq n, k\neq j$ we have:
\begin{IEEEeqnarray}{l}
G_{ij}(K_iC_{ij}) = I\label{gq5}\\
G_{ij}(K_iC_{ik}) = 0\label{gq6}
\end{IEEEeqnarray}
As the component of any $b_{kr}$ at $t_j\in T_{b_i}$ is zero if $k\neq i$, we have for $1\leq i,k\leq q, i\neq k, 1\leq j,r\leq n$ we have:
\begin{equation}
G_{ij}(K_iC_{kr} + R_iB_{kri}) = 0\label{gq7}
\end{equation}
We now consider the set $T_{a_2}$. Since the terminal $t_i\in T_{a_2}$ computes $a_i$ we have, for $1\leq i,j\leq n, j\neq i$
\begin{IEEEeqnarray}{l}
V_i(\sum_{k=1}^q U_kA_{ik}) = I\label{gq8}\\
V_i(\sum_{k=1}^q U_kA_{jk}) = 0\label{gq9}
\end{IEEEeqnarray}
The component of $b_{jk}$ is zero at $t_i\in T_{a_2}$ for $1\leq j\leq q, 1\leq i,k\leq n$, and hence we have,
\begin{equation}
V_i \big( ( \sum_{r=1,r\neq j}^q U_rB_{jkr}) + M_4D_{jk} \big) = 0\label{gq10}
\end{equation}
We now define the following matrices.
\begin{IEEEeqnarray}{l}
E = \begin{bmatrix} E_1 & E_2 & \cdots & E_n \end{bmatrix}^T\label{e}\\
M_2Q = \begin{bmatrix} M_2Q_1 & M_2Q_2 & \cdots & M_2Q_n \end{bmatrix}\label{m2q}\\
M_2C_{i} + M_3D_{i} = \begin{bmatrix} M_2C_{i1} + M_3D_{i1} \\ M_2C_{i2} + M_3D_{i2} \\ \vdots \\ M_2C_{in} + M_3D_{in} \end{bmatrix}^T\label{m2ci}\\
G_i = \begin{bmatrix} G_{i1} & G_{i2} & \cdots & G_{in} \end{bmatrix}^T\label{gi}\\
K_iQ + R_iA_{i} = \begin{bmatrix} K_iQ_1 + R_iA_{1i} \\ K_iQ_2 + R_iA_{2i} \\ \vdots \\ K_iQ_n + R_iA_{ni} \end{bmatrix}^T\label{kiq}\\
K_iC_{i} = \begin{bmatrix} K_iC_{i1} & K_iC_{i2} & \cdots & K_iC_{in} \end{bmatrix}\label{kici}\\
K_iC_{k} + R_iB_{ki} = \begin{bmatrix} K_iC_{k1} + R_iB_{k1i} \\ K_iC_{k2} + R_iB_{k2i} \\ \vdots \\ K_iC_{kn} + R_iB_{kni} \end{bmatrix}^T\label{kick}\\
V = \begin{bmatrix} V_1 & V_2 & \cdots & V_n \end{bmatrix}^T\label{v}\\
\sum_{k=1}^q U_kA_{k} = \begin{bmatrix} \sum_{k=1}^q U_kA_{1k} \\ \sum_{k=1}^q U_kA_{2k} \\ \vdots \\ \sum_{k=1}^q U_kA_{nk}\end{bmatrix}^T\label{sumukak}\\
 {(}\!\!\!\!\!\sum_{r=1,r\neq j}^q \!\! \!\! U_rB_{jr}) {+} M_4D_{j} = \! \begin{bmatrix} ( \sum_{r=1,r\neq j}^q U_rB_{j1r}) {+} M_4D_{j1} \\ ( \sum_{r=1,r\neq j}^q U_rB_{j2r}) {+} M_4D_{j2} \\ \vdots \\ ( \sum_{r=1,r\neq j}^q U_rB_{jnr}) {+} M_4D_{jn} \end{bmatrix}^T\!\!\!\!\! \IEEEeqnarraynumspace\label{urbjr}
\end{IEEEeqnarray}
Note that all of these newly defined matrices are square and are of size $dn \times dn$. For the rest of the paper let $I$ denote an identity matrix of size $dn\times dn$.
Applying Lemma~\ref{lem1} on equations (\ref{gq1}) and (\ref{gq2}) and using the newly defined matrices in equation (\ref{e}) and (\ref{m2q}) we get:
\begin{equation}
EM_2Q = I \label{em2q}
\end{equation}
From Corollary~\ref{coro1} and equations (\ref{gq3}), (\ref{e}) and (\ref{m2ci}) we get, for $1\leq i\leq q$:
\begin{equation}
E(M_2C_{i} + M_3D_{i}) = 0\label{em2ci}
\end{equation}
Similarly using Corollary~\ref{coro1} and equations (\ref{gq4}), (\ref{gi}) and (\ref{kiq}) we get, for $1\leq i\leq q$:
\begin{equation}
G_{i}(K_iQ + R_iA_{i}) = 0\label{gikiq}
\end{equation}
Using Lemma~\ref{lem1} and equations (\ref{gq5}), (\ref{gq6}), (\ref{gi}) and (\ref{kici}) we get, for $1\leq i\leq q$:
\begin{equation}
G_{i}(K_iC_{i}) = I\label{gikici}
\end{equation}
From Corollary~\ref{coro1} and equations (\ref{gq7}), (\ref{gi}) and (\ref{kick}) we have, for $1\leq i,k\leq n, i\neq k$:
\begin{equation}
G_{i}(K_iC_{k} + R_iB_{ki}) = 0\label{gikick}
\end{equation}
Employing Lemma~\ref{lem1} and equations (\ref{gq8}), (\ref{gq9}), (\ref{v}) and (\ref{sumukak}) we have:
\begin{equation}
V(\sum_{k=1}^q U_kA_{k}) = I\label{vsum}
\end{equation}
From Corollary~\ref{coro1} and equations (\ref{gq10}), (\ref{v}) and (\ref{urbjr}) we have, for $1\leq j\leq q$:
\begin{equation}
V\big( ( \sum_{r=1,r\neq j}^q U_rB_{jr}) + M_4D_{j} \big) = 0\label{vbigsum}
\end{equation}
The matrices $E, M_2$ and $Q$ are invertible from equation (\ref{em2q}). Matrix $G_i,K_i$ and $C_i$ for $1\leq i\leq q$ is invertible from equation (\ref{gikici}). And $V$ is invertible from equation (\ref{vsum}).
Since $E$ is invertible we have from equation (\ref{em2ci}):
\begin{equation}
M_2C_{i} + M_3D_{i} = 0 \label{1}
\end{equation}
 As both $M_2$ and $C_i$ are invertible matrices, their product is a full rank matrix, and hence from equation (\ref{1}), $M_3$ is an invertible matrix. This comes from the fact that a matrix of rank equal to a certain value cannot multiply with any other matrices and result in a matrix of rank equal to a greater value.
Since $G_i$ is invertible for $1\leq i\leq q$, we have from equation (\ref{gikiq}):
\begin{equation}
K_iQ + R_iA_{i} = 0 \label{2}
\end{equation}
Since both $K_i$ and $Q$ are invertible matrices, their product is a full rank matrix, and hence $R_i$ is an invertible matrix for $1\leq i\leq q$.
Also from equation (\ref{gikick}) we have, for $1\leq i,k\leq q, i\neq k$:
\begin{equation}
K_iC_{k} + R_iB_{ki} = 0 \label{3}
\end{equation}
And since $V$ is invertible we have from equation (\ref{vbigsum}), for $1\leq i\leq q$:
\begin{equation}
( \sum_{r=1,r\neq i}^p U_rB_{ir}) + M_4D_{i} = 0 \label{4}
\end{equation}
Substituting $D_i$ from equation (\ref{1}) in equation (\ref{4}) we get, for $1\leq i\leq q$:
\begin{IEEEeqnarray*}{l}
( \sum_{r=1,r\neq i}^q U_rB_{ir}) - M_4M_3^{-1}M_2C_{i} = 0\\
\text{Substituting $B_{ir}$ from equation (\ref{3})} we get:\\
-( \sum_{r=1,r\neq i}^q\!\! U_rR_r^{-1}K_rC_i) - M_4M_3^{-1}M_2C_{i} = 0\\
\text{Substituting $R_r^{-1}K_r$ from equation (\ref{2})} we get:\\
( \sum_{r=1,r\neq i}^q U_rA_rQ^{-1}C_i) - M_4M_3^{-1}M_2C_{i} = 0\\
\text{Substituting $Q^{-1}$ from equation (\ref{em2q}) we get:}\\
( \sum_{r=1,r\neq i}^q U_rA_rEM_2C_i) - M_4M_3^{-1}M_2C_{i} = 0\\
\text{or, }\big(( \sum_{r=1,r\neq i}^q U_rA_rE) - M_4M_3^{-1}\big)M_2C_{i} = 0\\
\text{Since $M_2$ and $C_i$ both are invertible}:\\
( \sum_{r=1,r\neq i}^q U_rA_rE) - M_4M_3^{-1} = 0\\
\text{Substituting $\sum_{r=1,r\neq i}^q U_rA_r$ from equation (\ref{vsum}) we get}:\\
(V^{-1} -U_iA_i)E - M_4M_3^{-1} = 0 \\
\text{or, }V^{-1} -U_iA_i = M_4M_3^{-1}E^{-1}\\
\text{or, }U_iA_i = V^{-1} - M_4M_3^{-1}E^{-1}\IEEEyesnumber\label{5}
\end{IEEEeqnarray*}

Now, substituting equation (\ref{5}) in equation (\ref{vsum}) we get
\begin{IEEEeqnarray*}{l}
V(\sum_{i=1}^q V^{-1} - M_4M_3^{-1}E^{-1}) = I\\
\sum_{i=1}^q V(V^{-1} - M_4M_3^{-1}E^{-1}) = I\\
\sum_{i=1}^q I - VM_4M_3^{-1}E^{-1} = I\\
(q-1)I = qIVM_4M_3^{-1}E^{-1}\IEEEyesnumber\label{6}
\end{IEEEeqnarray*}
In equation (\ref{6}), if $q=0$, then the equation becomes $-I = 0$. Hence $q\neq 0$ is a necessary condition for the network $\mathcal{N}_2$ to have a rate $\frac{1}{n}$ fractional linear network coding solution. Then from the fact that an element in a finite field is equal to zero if and only if the characteristic of the finite field divides that element, it can be concluded that $q\neq 0$ if and only if the characteristic of the finite field does not divides $q$.

We now show that $\mathcal{N}_2$ has a $(1,n)$ fractional linear network coding solution if the $q$ has an inverse in the finite field. Note that, as discussed above, $q$ has an inverse if and only if the characteristic of the finite field does not divides $q$. Let an $n$-length vector whose $i^{\text{th}}$ component is $a_i$ and all other components are zero be denoted by $\bar{a}_i$. Also let $n$-length vector whose $j^{\text{th}}$ component is $b_{ij}$ and all other components are zero be denoted by the notation $\bar{b}_{ij}$. Note that $a_j$ and $b_{ij}$ for $1\leq i\leq q, 1\leq j\leq n$ are the source processes. Then it can be seen that for proper local coding matrices the following information can be transmitted by the corresponding edges.
\begin{IEEEeqnarray*}{l}
Y_{e_a} = \sum_{j=1}^n \bar{a}_j + \sum_{i=1}^q \sum_{j=1}^n \bar{b}_{ij}\\
\text{for } 1\leq i\leq q: \quad Y_{e_i} = \sum_{j=1}^n \bar{a}_j + \sum_{k=1,k\neq i}^q \sum_{j=1}^n \bar{b}_{kj}\\
Y_{e_b} = \sum_{i=1}^q \sum_{j=1}^n \bar{b}_{ij}\\
Y_{e_a^\prime} = Y_{e_a} - Y_{e_b} = \sum_{j=1}^n \bar{a}_j\\
\text{for } 1\leq i\leq q: \quad Y_{e_i^\prime} = Y_{e_a} - Y_{e_i} = \sum_{j=1}^n \bar{b}_{ij}\\
Y_{e_b^\prime} = q^{-1}\{\sum_{i=1}^q Y_{e_i} - (q-1)Y_{e_b}\} = q^{-1}\{q(\sum_{j=1}^n \bar{a}_j) \\+\> (q-1)(\sum_{i=1}^q \sum_{j=1}^n \bar{b}_{ij}) - (q-1)(\sum_{i=1}^q \sum_{j=1}^n \bar{b}_{ij})\} = \sum_{j=1}^n \bar{a}_j
\end{IEEEeqnarray*}
Let $\check{u}(j)$ be a unit row vector of length $n$ which has $j^{\text{th}}$ component equal to one and all other component equal to zero. Then from the dot product of $\check{u}(j)$ and $\sum_{j=1}^n \bar{a}_j$, message $a_j$ can be retrieved. Similarly from the dot product of $\check{u}(j)$ and $\sum_{j=1}^n \bar{b}_{ij}$, $b_{ij}$ can be determined.
\end{IEEEproof}
\begin{theorem}\label{thm2}
For any non-zero positive rational number $\frac{k}{n}$ and for any finite set of prime numbers $\{p_1,p_2,\ldots,p_l\}$ there exists a network which has a rate $\frac{k}{n}$ fractional linear network coding solution if and only if the characteristic of the finite field does not belong to the given set of primes.
\end{theorem}
\begin{IEEEproof}
Let $q$ be equal to $p_1.p_2.\ldots .p_l$ in $\mathcal{N}_2$. Let us construct $\mathcal{N}_2^\prime$ by joining $n$ copies of $\mathcal{N}_2$ at the corresponding sources and the terminals, in a similar way $\mathcal{N}_1^\prime$ was constructed from $\mathcal{N}_1$. It can be also seen that Lemma~\ref{lem2} holds true when $\mathcal{N}_1$ and $\mathcal{N}_1^\prime$ are replaced by $\mathcal{N}_2$ and $\mathcal{N}_2^\prime$ respectively. So if $\mathcal{N}_2^\prime$ has a $(k,n)$ fractional linear network coding solution then $\mathcal{N}_2$ has a $(k,kn)$ fractional linear network coding solution.

Now say $\mathcal{N}_2^\prime$ has a rate $\frac{k}{n}$ fractional linear network coding solution even if the characteristic of the finite belongs to the set $\{p_1,p_2,\ldots,p_l\}$. Since $q = p_1.p_2.\ldots .p_l$, note that $q=0$ over such a field. Then, $\mathcal{N}_2$ has a rate $\frac{k}{kn} = \frac{1}{n}$ fractional linear network coding solution over a finite field in which $q=0$. However, this is in contradiction with Lemma~\ref{lem4}.

If however, the characteristic does not belong to the given set of primes, then, since there are $n$ copies of $\mathcal{N}_2$ in $\mathcal{N}_2^\prime$, and each copy has a $(1,n)$ fractional linear network coding solution, a $(k,n)$ fractional linear network coding solution can easily be constructed for $\mathcal{N}_2^\prime$.
\end{IEEEproof}

\section{A multiple-unicast network having a rate $\frac{k}{n}$ fractional linear network coding solution iff the characteristic belongs to a given finite/co-finite set of primes}\label{sec4}
\begin{figure}
\centering
\includegraphics[width=0.48\textwidth]{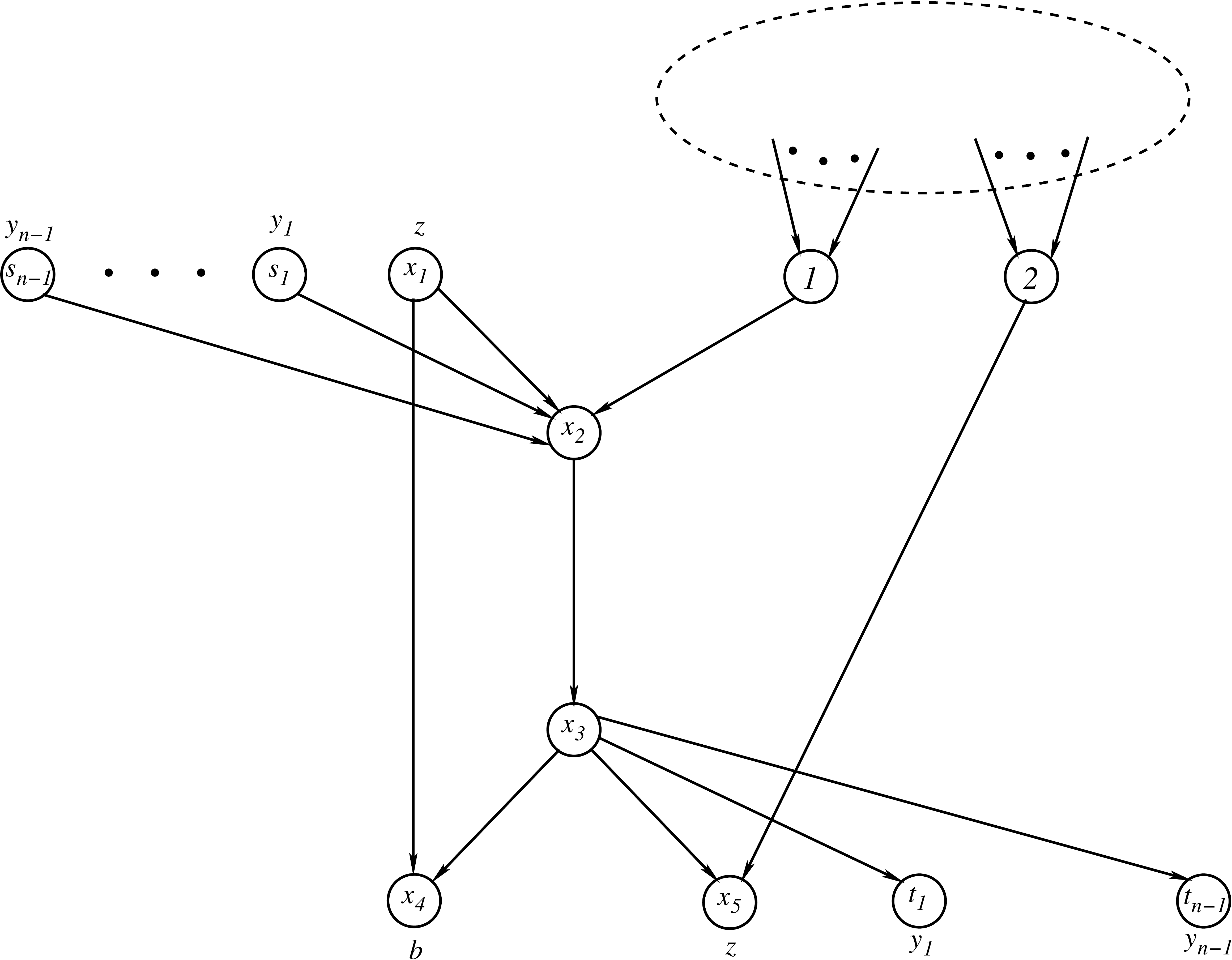}
\caption{Gadget which attaches to the two terminals of any arbitrary network indicated by the dotted lines. Nodes $s_1,\ldots ,s_{n-1}$ are sources and source $s_i$ generates the messages $y_i$ for $1\leq i\leq (n-1)$. The nodes $t_1, \ldots ,t_{n-1}$ are terminals and $t_i$ demands $y_{i}$ for $1\leq i\leq (n-1)$. The nodes $n_1$ and $n_2$ were terminals in the original network (network indicated by the dotted line) and both of these nodes demanded the message $b$. The rest of the naming convention has been kept the same as it was in \cite{non}.}
\label{nw}
\end{figure}
In this section we show that for any non-zero positive rational number $\frac{k}{n}$ and for any finite/co-finite set of primes, there exits a multiple-unicast network which has a rate $\frac{k}{n}$ fractional linear network coding solution if and only if the characteristic of the finite field belongs to the given set. To prove this result, we first show that for each of the networks $\mathcal{N}_1$ and $\mathcal{N}_2$ presented in Section~\ref{sec3}, there exists a multiple-unicast network which has a $(1,n)$ fractional linear network coding solution if and only if the the corresponding network $\mathcal{N}_1$ or $\mathcal{N}_2$ has a $(1,n)$ fractional linear network coding solution.

In a multiple-unicast network, by definition, each source process is generated at only one source node and is demanded by only one terminal. Additionally, each source node generates only one source process, and each terminal demands only one source process. In both the networks $\mathcal{N}_1$ and $\mathcal{N}_2$ there exists no source processes which is generated by more than one source node, or no source node generates more than one source process. Moreover, there does not exist any terminal which demands more than one source process. However, there exists more than one terminal which demands the same source process. This is fixed in the following way.

In \cite{non} it was shown that for any network there exists a solvably equivalent multiple-unicast network. To resolve the case of more than one terminals demanding the same source message, the authors considered two such terminals at a time and added a gadget to the two terminals. The same procedure is followed here, only the gadget has been modified. This modified gadget is shown in Fig.~\ref{nw}. It is assumed that the nodes $n_1$ and $n_2$ both demanded the same message $b$ in the original network (network before attaching the gadget). After adding the gadget, the modified network has $n$ more source nodes $x_1,s_1,\ldots ,s_{n-1}$, and $n+1$ new terminal nodes $x_4,x_5,t_1,\ldots ,t_{n-1}$. Nodes $n_1$ and $n_2$ becomes intermediate nodes in the modified construction. This process has to be repeated iteratively for every two terminals in the original network that demand the same source process. In the same way as shown in \textit{Theorem II.1} of \cite{non}, it can be shown that after the completion of this process, the resulting network has a $(1,n)$ fractional linear network coding solution if and only if the original network has a $(1,n)$ fractional linear network coding solution.

Hence, as shown above, corresponding to each of the networks $\mathcal{N}_1$ and $\mathcal{N}_2$, there exist multiple-unicast networks $\mathcal{N}_1^m$ and $\mathcal{N}_2^m$ which have a $(1,n)$ fractional linear network coding solution if and only if $\mathcal{N}_1$ and $\mathcal{N}_2$ have a $(1,n)$ fractional linear network coding solution respectively. Now by connecting $k$ copies of $\mathcal{N}_1^m$ and $\mathcal{N}_2^m$ in the same way as $\mathcal{N}_1^\prime$ and $\mathcal{N}_2^\prime$ was constructed from $\mathcal{N}_1$ and $\mathcal{N}_2$ respectively, the following theorem can be proved in a similar way to Theorem~\ref{thm1} and Theorem~\ref{thm2}.
\begin{theorem}
For any non-zero positive rational number $\frac{k}{n}$ and for any finite/co-finite set of prime numbers there exists a multiple-unicast network which has a rate $\frac{k}{n}$ fractional linear network coding solution if and only if the characteristic of the finite field belongs to the given set of primes.
\end{theorem}

\section{Conclusion}\label{sec5}
In this paper we have shown that for any non-zero positive rational number $\frac{k}{n}$ and any finite/co-finite set of prime numbers there exists a multiple unicast network which has a rate $\frac{k}{n}$ fractional linear network coding solution if and only if the characteristic of the finite field belongs to the given set. To prove the existence, we have explicitly presented networks having desired properties. The generalized Fano and generalized non-Fano networks presented in \cite{gf} are special cases of the networks presented in this paper.

\end{document}